\documentclass[aps,prl,twocolumn,letterpaper]{revtex4-1}
\usepackage[utf8]{inputenc}

\makeatletter
\let\warning\@latex@warning
\makeatother

\usepackage{amsmath,amssymb,amsthm,mathtools,bm}
\newcommand{\yesnumber}{\refstepcounter{equation}\tag{\theequation}}
\usepackage{physics}
\newcommand{\projector}[1]{\ketbra{#1}{#1}}

\newcommand{\eqq}{\equiv}
\newcommand*{\eps}{\varepsilon}
\DeclareMathOperator{\id}{id}

\usepackage{xcolor}
\usepackage{hyperref}
\hypersetup{%
    breaklinks,
    pdfpagemode=UseNone,
    pdfstartview=FitH,
    colorlinks,
    linkcolor={red!50!black},
    citecolor={red!50!black},
    urlcolor={red!30!black}
}
\usepackage[capitalize]{cleveref}

\usepackage{graphicx}
\graphicspath{{figures/}{./}}
\usepackage{tikz}
\usepackage{qcircuit}
\newcommand{\puremultigate}[2]{*+<1em,.9em>{\hphantom{#2}} \POS [0,0]="i",[0,0].[#1,0]="e",!C *{#2},"e"+UR;"e"+UL **\dir{-};"e"+DL **\dir{-};"e"+DR **\dir{-};"e"+UR **\dir{-},"i"}

\newcommand{\dw}{\ar @{.} [0,-1]}

\let\oldvartheta\vartheta
\let\vartheta\theta
\let\theta\oldvartheta 
\renewcommand{\lim}{\qopname\relax m{lim\vphantom{p}}}
\renewcommand{\inf}{\qopname\relax m{inf\vphantom{p}}}
\makeatletter

\renewcommand{\limsup}{\mathop {\mathpalette \varlimsup@ {\vphantom{p}}}\nmlimits@}
\makeatother

\let\olddownarrow\downarrow
\renewcommand{\downarrow}{{\scriptscriptstyle \olddownarrow}}

\theoremstyle{plain}
\newtheorem {theorem}{Theorem} 
\newtheorem*{theorem*}       {Theorem}
\newtheorem {lemma}[theorem]{Lemma} 
\newtheorem*{lemma*}      {Lemma}
\newtheorem {corollary}[theorem]{Corollary} 
\newtheorem*{corollary*}      {Corollary}
\newtheorem*{definition*}      {Definition}

\theoremstyle{remark}

\newtheorem*{remark*}{Remark}

\newtheorem*{note*}{Note}
\newtheorem {observation}[theorem]{Observation} 
\newtheorem*{observation*}      {Observation}

\begin{document}

\title{Upper bounds on device-independent quantum key distribution}
\author{Matthias Christandl}
\email{christandl@math.ku.dk}
\affiliation{QMATH, 
    Department of Mathematical Sciences, 
    University of Copenhagen, 
    Universitetsparken 5, 2100 Copenhagen \O, Denmark}
\author{Roberto Ferrara}    
\email{roberto.ferrara@tum.de}
\affiliation{Lehr- und Forschungseinheit f\"ur Nachrichtentechnik,
Technische Universit\"at M\"unchen, 80339 Munich, Germany}
\author{Karol Horodecki}
\email{khorodec@inf.ug.edu.pl}
\affiliation{Institute of Informatics, 
	National Quantum Information Centre, Faculty of Mathematics, 
	Physics and Informatics, University of Gda{\'n}sk, 80-308 Gda{\'n}sk, Poland}
	\affiliation{International Centre for Theory of Quantum Technologies,
University of Gda{\'n}sk, 80-952 Gda{\'n}sk, Poland}

\begin{abstract}
Quantum key distribution (QKD) is a method that distributes a secret key to a sender and a receiver by the transmission of quantum particles (e.g.
photons). 
Device-independent quantum key distribution (DIQKD) is a version of QKD with a stronger notion of security, in that the sender and receiver base their protocol only on the statistics of input and outputs of their devices as inspired by Bell's theorem.
We study the rate at which DIQKD can be carried out for a given bipartite quantum state distributed between the sender and receiver or a quantum channel connecting them. 
We provide upper bounds on the achievable rate going beyond upper bounds possible for QKD. In particular, we construct states and channels where the QKD rate is significant while the DIQKD rate is negligible. This gap is illustrated for a practical case arising when using standard post-processing techniques for entangled two-qubit states. 
\end{abstract}
\date{\today}
\maketitle


\paragraph{Introduction.}

Quantum key distribution (QKD) offers the possibility to distribute a perfectly secure key among two parties via quantum communication~\cite{BB84}. 
The parties can later use this key for perfectly secure communication. 
Whereas theoretically, the security of QKD is very well understood, the experimental implementations remain challenging. 
This is because it is difficult to verify that the theoretical models and the experimental implementations fit together. 
In practice, the exploitation of differences between model and implementation are known as side channels, and it is here that quantum communication opens a can of worms~\cite{Makarov_2009}. 
QKD is thus dependent on the exact known specifications of the devices used: it is \emph{device-dependent}.

Ekert's scheme for QKD merely verified by the violation of a Bell inequality opens up the possibility of \emph{device-independent} QKD (DIQKD), as the violation of a Bell inequality can be argued by the obtained correlations alone (under the assumption of appropriate timing of the signals). 
In recent years, DIQKD has been put on a firm footing~\cite{E91,BHK,QDICollective,DIMayersYao,Vazirani-Vidick,ArnonFriedman2018,TowardsDIQKD,LHFB-Hensen2015,LHFB-Giustina2015,LHFB-Shalm2015, ArnonFriedman2018}. 
However, it should be said that a device purchased from an adversarial vendor emphasizes other types of loopholes, for instance, the hidden storage and later unwanted release of the bits generated~\cite{BCK13,CL19,Bennettprivatecomm}.

Whereas security proofs for both QKD and DIQKD are involved, e.g., since channel noise needs to be estimated and the eavesdropper might carry out non-i.i.d.\ attacks, upper bounds on the optimal rate can be obtained with a Shannon-theoretic approach. 
In the case of QKD, the corresponding rates are the key rate $K(\rho)$ of a bipartite state $\rho$ shared among the communicating parties~\cite{DW05,HHHO05} and the private capacity $\mathcal{P}(\Lambda)$ of a quantum channel $\Lambda$~\cite{Devetak2005}. 
Interestingly, these rates can indeed be achieved in the actual QKD setting, e.g.~by use of the post-selection technique~\cite{Christandl2009}. For the first upper bounds on DIQKD rates see \cite{KWW20}.

In this paper, we consider the natural DIQKD variants $K^\text{DI}(\rho)$ and $\mathcal{P}^\text{DI}(\Lambda)$. 
Since DIQKD has a higher security demand than QKD, one has the trivial bounds
$K^\text{DI}(\rho)\leq K(\rho)$ and $\mathcal{P}^\text{DI}(\Lambda) \leq \mathcal{P}(\Lambda)$.

Our main results are upper bounds on the DIQKD rates that go beyond the bounds implied by QKD, thereby pointing out a fundamental difference between standard and device-independent QKD. We illustrate the bounds with an example where the QKD rate is constant but, remarkably, the DIQKD rate is vanishingly low. 
We will also discuss a practical example with an explicit gap. 
In the following we introduce the setting before presenting and illustrating the main results.

\paragraph{Communication rates in quantum cryptography.}
\hfill\\
Every QKD and DIQKD protocol consists of preparing, exchanging and measuring quantum particles, followed by the post-processing of the measurement data resulting in the final key. Note that these are not necessarily separate stages, but may be interwoven. 
Most QKD protocols, however, can be modeled as an establishment of $n$ independent copies of a bipartite quantum state $\rho$ between the communicating parties Alice and Bob, followed by a protocol consisting of local operations and public communication (LOPC).
For simplicity, we will assume that all Hilbert spaces are finite-dimensional. 
This protocol results in a final key secret against an eavesdropper holding the purification of $\rho^{\otimes n}$ and a copy of all classical communication. 
When maximizing over possible LOPC protocols, one obtains the key rate $K(\rho)$.

If Alice and Bob have control over their measurement apparatuses, there exist effective methods to verify that they indeed have $n$ independent copies of $\rho$, even if the adversary interferes with the quantum communication. 
Thus $K(\rho)$ also has the practical relevance as a QKD rate and not only information-theoretic meaning~\cite{Christandl2009}.

Instead of modeling the distribution of the quantum particles by a density matrix $\rho$, one might also model it as arising from a quantum channel $\Lambda$, a completely-positive trace-preserving linear map. 
This scenario, which results in the private capacity $\mathcal{P}(\Lambda)$ is more general but more cumbersome to treat. 
Therefore, we will focus on the density matrix paradigm, yet also state our results in the channel paradigm. 

Note that in most practical protocols, in QKD but especially in DIQKD, measurements are performed on single copies of $\rho$ by POVMs $\{A_a^x\}_a$ and $\{B_b^y\}_b$. 
We denote the measurement choices by $x$ and $y$, respectively, and the outcomes by $a$ and $b$. 
If an eavesdropper does not interfere with the measurement, this results in $n$ independent and identical samples of the conditional probability distribution 
\[p(a, b|x, y) \coloneqq \tr [(A_a^x \otimes B_b^y) \rho].\]
Classical post-processing then leads to the final secret key against an eavesdropper who holds the purification of the state $\rho^{\otimes n}$ as well as a transcript of all public communication. 
We note that the distribution of the measurement choice $p(x, y)=p(x)p(y)$ is usually fixed (e.g.~uniform) so that the samples are actually drawn from the distribution $p(a, b, x, y)=p(a, b|x, y) p(x,y)$, rather than from $p(a, b|x, y)$. 
The choice of measurements and their distribution is denoted by $\mathcal{M}$. 
We denote the corresponding QKD rate when maximizing over POVMs by $K^{(1)}(\rho)$, indicating that the measurement acts on one copy of the state. Note that 
\[K^{(1)}(\rho) \leq K(\rho).\]

In DIQKD, in contrast to QKD, Alice and Bob know neither the measurement operators performed by their apparatus nor the states measured. 
In particular, even though they can verify that they have $n$ independent copies of $p(a, b|x, y)$, it might not be possible to infer that the underlying quantum process respects the independent nature. Namely, it might not be possible to prove that the measurements $\{A_a^x\}_a$ and $\{B_b^y\}_b$ were indeed carried out independently on independent copies of $\rho$, rather than some more complicated procedure. 
Even assuming that the device indeed performed $n$ identical independent measurements on an identical quantum state, leading to what we call the DIQKD rate
\[K^\text{DI}(\rho)\] 
leaves open the possibility for different measurements as we will explain in the following. We emphasize that it is possible, yet unproven, that this rate can be achieved in a realistic DIQKD setting, as recent research indicates~\cite{ArnonFriedman2018,Rotem-phd} (cf. research on quantum de Finetti theorems \cite{Christandl2007one, renner2005security, Renner07}). 
Since knowing less about the apparatus can only decrease the rate, we have 
\begin{equation}
    K^\text{DI}(\rho)\leq K^{(1)}(\rho)\leq K(\rho).
    \label{eq:singlecopy-povm}
\end{equation}
In the following, we will provide upper bounds that improve on this bound and exploit them to present a gap between $K^\text{DI}(\rho)$ and $K(\rho)$.

\paragraph{Upper bounds on DIQKD.}

Assume now that the POVMs $\{A_a^x\}_a$ and $\{B_b^y\}_b$ are optimal for $K^\text{DI}(\rho)$ (such POVMs exist by compactness, since the Hilbert spaces are finite-dimensional). 
Note that there might exist a different state $\rho'$ and different measurements $\{{A'}_a^x\}_a$ and $\{{B'}_b^y\}_b$ leading to the same distribution
\[p(a, b|x, y) \coloneqq \tr [(A_a^x \otimes B_b^y) \rho]= \tr [({A'}_a^x \otimes {B'}_b^y) \rho'].\]
In this case, we write $(\mathcal{M}, \rho)\equiv (\mathcal{M}', \rho')$. 
We thus see that the maximal achievable key rate for $\rho$ is also achievable for $\rho'$.
We thus have
\[K^\text{DI}(\rho)\leq K^\text{DI}(\rho').\]

Combining this bound with \cref{eq:singlecopy-povm} we find that 
\begin{align}
    \label{eq:di-upperbound}
    K^\text{DI}(\rho)\leq \sup_{\mathcal{M}} \inf_{\substack{(\mathcal{M}', \rho'): \\ (\mathcal{M}, \rho) \equiv (\mathcal{M}', \rho')}} K(\rho').
\end{align}
A proof based on the formal definitions of the involved rates is given in the Supplementary Material.

We will now give a construction of examples, where $(\mathcal{M}, \rho)\equiv (\mathcal{M}', \rho')$. For this, note that transposing Bob's system does not change the probabilities
\[\tr [(A_a^x \otimes B_b^y) \rho]= \tr [(A_a^x \otimes (B_b^y)^T) \rho^\Gamma].\]
Here, $T$ denotes the transpose and $\Gamma$ the partial transpose. 
The density matrix $\rho$ can lose the property of being positive semi-definite  after partial transposition. 
For the equation above to be valid, we thus require $\rho^{\Gamma} \geq 0$, in which case $\rho$ is said to be PPT (Positive under Partial Transposition). 
PPT states are the only known examples of {bound-entangled} states, that is, entangled states from which no pure entanglement can be extracted at a positive rate~\cite{HHH98}.
Still, they form a rich class of states, including states from which a secret key can be extracted at positive rates~\cite{HHHO05,HHHO09} (similar results are known for channels~\cite{HHHLO08,UncondPRL,AlmostZeroChannel}). 
There are even examples of PPT entangled states that violate Bell inequalities~\cite{Vrtesi2014}. 
When restricting to PPT states $\rho$, we therefore find
\begin{equation}
\label{eq:pptupperbound}
K^\text{DI}(\rho) \leq \min\{K(\rho^{\Gamma}),K(\rho)\}.
\end{equation}
To see the significance of the above result, it is important to note, that there are PPT states for which $K(\rho)$ is high, but $K(\rho^{\Gamma})$ is low~\cite{BCHW15,HHHO05}. 
This implieas a gap via the above inequality and therefore a fundamental difference between device-dependent and device-independent secrecy.

We now provide an example of states exhibiting this gap. 
Aiming at constructions with relatively few qubits, we further develop the results of~\cite{BCHW15,HPHH08} (see also the Supplementary Material).
In general, this gap holds for all those examples of PPT states that are close to private bits,
but that after partial transposition become close to separable states~\cite{HPHH08,HHHO05,karol-PhD,HM15,BCHW15}.

\paragraph{Examples.}

We consider the $2d \times 2d$ dimensional states from~\cite{BCHW15} which are of the form
\begin{align*}
\rho_d &\coloneqq \frac{1}{2}
\begin{bmatrix}
	(1-p)\sqrt{XX^{\dagger}} & 0 & 0 & (1-p) X \\
	0 & pY &  0 & 0 \\
	0 &  0 & pY & 0 \\
	(1-p) X^{\dagger} & 0 & 0 & (1-p) \sqrt{X^{\dagger}X}
\end{bmatrix},
\end{align*}
with $X$ and $Y$ to be chosen later, satisfying $\|X\|_1 = \|Y\|_1=1$.
The qubit systems are called the \emph{key} systems and the qudits are called the \emph{shield} systems.
By the privacy-squeezing technique of~\cite{HHHO09}, this state has at least as much key as the key obtained by measuring 
\begin{equation*}
\rho_\mathrm{Bell} \coloneqq \frac{1}{2}
\begin{bmatrix}
	(1-p) & 0 & 0 & (1-p) \\
	0 & p &  0 & 0 \\
	0 &  0 & p & 0 \\
	(1-p) & 0 & 0 & (1-p)
\end{bmatrix},
\end{equation*}
which is a Bell diagonal state.
A lower bound on this key is given by the Devetak-Winter protocol~\cite{DW05},
which was also derived in~\cite[Eq.~(22)]{ABBBMM06} and reads
\begin{equation*}
K_D(\rho_\mathrm{Bell})  \geq  1 - H\left( \left(1-p,\frac{p}{2},\frac{p}{2} \right) \right),
\end{equation*}
where $H$ is the Shannon entropy. 

In order for $\rho_d$ to be PPT, we choose  $Y = \frac{1}{d} \sum_{i=0}^{d-1} \projector{ii}$ and  $X = 1/(d \sqrt{d})\sum_{i,j=0}^{d-1}u_{ij} \ketbra{ij}{ji}$, with $u_{ij}$ being complex numbers of modulus $\frac{1}{ \sqrt{d}}$ such that $U=\sum_{ij}u_{ij}\ketbra{i}{j}$ is a unitary matrix~\cite{HPHH08}. 
In particular, one can take $U$ to be the Fourier transform or (if $d$ is a power of two) a tensor power of the Hadamard matrix. 
We also choose $p = \frac{1}{ \sqrt{d} +1}$.
To conclude, we derived the lower bound $K(\rho_d) \geq 1- H\left(\left(\sqrt{d}, \frac12, \frac12\right)/(\sqrt{d}+1)\right)$, while the upper bound $K(\rho_d^\Gamma) \leq \frac{1}{\sqrt{d}+1}$  was computed as part of~\cite[Supplementary material, Corollary~40]{ChristandlFerrara}. See also Theorem 2 in the Supplemental Material, where \cite{lock-ent} is used.

A quick check reveals that $K(\rho_d) > K(\rho_d^\Gamma)$ for all $d\geq 24$, i.e.~for all states $\rho_d$ with at least three qubits and a qutrit in the shield at each side.
In particular, $\rho_{2^{5}}$ is thus a $12$ qubit state, which proves the separation between the device-dependent and the device-independent key.
For $20$ qubits of shield per side, we arrive at a state which has $K(\rho_{2^{20}})\geq 0.98$ and $K(\rho_{2^{20}}^\Gamma) \leq 1/(2^{10} +1)\sim 10^{-3}$.

\begin{remark*}
At first, this does not seem to be a practical example. Note, however that using the common subroutine advantage distillation on $\rho_d^{\otimes n}$ yields the same lower and upper bounds as $\rho_{d^n}$. Our results thus directly concern the amount of key distilled after advantage distillation \cite{Maurer93} on the key part of 20 copies of $\rho_2$ if we make sure that the other 20 qubits of shield do not get in the hands of the eavesdropper. 
In particular, we see that whereas in QKD, the obtained bit in this setting is secure, the upper bound tells us that this bit is not secure in a device-independent setting. 
Therefore the state, and particularly any of its parts, including the shield, cannot be tested independently of the device.
The quantum operation of removing a system (in our case, the shield) from the reach of the eavesdropper is based on trust in the quantum memories and cannot be certified by classical correlations alone.
\end{remark*}

\paragraph{Device-Independent Entanglement Measures.}
Implicit in the upper bound on $K(\rho^\Gamma)$ was the use of the relative entropy of entanglement $E_\mathrm{r}$.
In this context, it is therefore natural to introduce \emph{device-independent entanglement measures}. 
In analogy to \cref{eq:di-upperbound}, for any entanglement measure $E$ we define 
\begin{equation}
E^{\downarrow}(\rho)\coloneqq\sup_\mathcal{M} \inf_{(\mathcal{N},\sigma)\equiv (\mathcal{M},\rho)} E(\sigma) \leq E(\rho),
\end{equation}
where we use the down arrow to indicate the optimization over Eve's implementation of the device, in close analogy to the down arrow used in the intrinsic information~\cite{MW99}, where also an optimization over Eve’s action is carried out.
Notice that $E^{\downarrow}(\projector\psi) = E(\projector\psi)$ because all pure states are self-testable~\cite{coladangelo2017all}.
If $E$ is either the distillable key $K$ or an upper bound on it, it then follows that
\begin{equation}
 K^\text{DI}\leq {K}^{\downarrow} \leq {E}^{\downarrow}\leq E.
\label{eq:forall-sq-bound}
\end{equation}
In particular, for $E$ being the squashed entanglement $E_\mathrm{sq}$ or the relative entropy of entanglement $E_\mathrm{r}$, we obtain 
\begin{align*}
K^\text{DI} \leq \min\{E_\mathrm{sq}^{\downarrow}, E_\mathrm{r}^{\downarrow}\}.
\end{align*}
In the example above, the relative entropy bound was implicitly used together with
$E_\mathrm{r}^{\downarrow}(\rho) \leq E_\mathrm{r}(\rho^\Gamma)$ for PPT states $\rho$.
Note that fixing a choice of ${\cal M}$ in $E^{\downarrow}$ also produces a device-independent entanglement measure of a distribution.

\paragraph*{Device-Independent Private Capacity.}
The ideas presented so far can also be applied to the private capacity $\mathcal{P}(\Lambda)$ of a channel $\Lambda$. 
They are thus useful in the most general setting, where, for instance, the optical fiber itself is modelled and not only the states produced when using the optical fiber. 

There are different natural versions of the private capacity depending on whether assistance by public communication is restricted to being one-way ($\mathcal{P}_1$) or whether general two-way communication is allowed ($\mathcal{P}_2$). In the information-theoretic setting, the setting without publication communication ($\mathcal{P}_0$) is also meaningful. 
With increased power comes increased rate, and thus $$\mathcal{P}_0 \le \mathcal{P}_1 \le \mathcal{P}_2.$$

The device-independent private capacity also has three analogous versions ${\mathcal{P}}_i^\text{DI},\ i=0,1,2$ corresponding to whether two-way, one-way or no public communication is given to Alice and Bob \emph{outside} the devices. 
Additionally, there will be different classes of adversarial devices, depending on whether we consider adversaries that, besides the quantum channel, use two-, one- or no-way public communication \emph{inside} the devices to produce the state to be measured.
Arguably, allowing less classical communication in the device than the one used by Alice and Bob is physically unsound, but can be used as a mathematical tool to reach some upper bounds.
Thus, we can restrict ourselves to adversarial devices that use no public communication, which can only make the rates larger.
Similarly, we also consider i.i.d.\ devices that do not use memory between the input states of different channel uses.
Again, these are not realistic implementations of a device delivered by an adversary but merely a tool to provide upper bounds.
Indeed, in practical scenarios the provided devices will often be from a cooperating rather than an adversarial party.
These devices will use quantum memories at Alice and Bob and even classical communication \emph{outside} the classical input-output rounds where communication is allowed, to maximize the key.
In the Supplementary Material, we explore the various rates obtained when considering different classes of devices allowed to the adversary and the different variants of public communications that are allowed to the intended parties.

We now introduce the class of i.i.d.\ devices that use neither public communication nor memory between channel uses.
A device for a channel $\Lambda$ from Alice to Bob is given by a tuple $(\mathcal{M}, \rho, \Lambda)$ of measurements $\mathcal{M}$ on Alice and Bob's side, a bipartite state $\rho$ (half of which is the input to the channel), and a channel $\Lambda$. 
The conditional probability distribution is then obtained, as shown in \cref{fig:IDI0}, via
\begin{equation*}
    p(ab|xy) = \tr[(\id \otimes \Lambda)(\rho) \cdot M_a^x\otimes M_b^y].
\end{equation*}
We again write $(\mathcal{N},\sigma, \Lambda' )\equiv(\mathcal{M},\rho, \Lambda)$ for devices that produce the same distribution.
As in the case of entanglement measures for states, we can use any channel entanglement measure $\mathcal{E}(\Lambda)$ to define a device-independent version 
\begin{equation}
\mathcal{E}^{\downarrow}(\Lambda) \equiv \mathcal{E}^{\downarrow_0}(\Lambda)
\coloneqq 
\sup_{\mathcal{M},\rho} \,
\inf_{{%
    {(\mathcal{N},\sigma, \Lambda' )\equiv(\mathcal{M},\rho, \Lambda) }
    }} 
\mathcal{E}(\Lambda')    
\end{equation}
(see~\cite{TGW14bound,Pirandola2017} for the channel generalizations of $E_\mathrm{sq}$ and $E_\mathrm{r}$ respectively, as well as~\cite{christandl2017relative} for the use of the latter). See also \cite{SRS08}.

\begin{figure}
    \begin{gather*}
    \Qcircuit @C=.5em @R=.3em{
        &&& &\push{x_1}&\puremultigate{1}{\ }\cw&\push{\,a_1}\cw&
        &&&&&\push{x_2}&\puremultigate{1}{\ }\cw&\push{\,a_2}\cw&
        \\
        &&& & \qw[-3]&\ghost{\ }&&
        &&&&& \qw[-4]&\ghost{\ }&&
        \\
        \push{\sigma} \ar@{-}[1,1]\ar@{-}[-1,1]&&&&&&&
        \push{\sigma} \ar@{-}[1,1]\ar@{-}[-1,1]&&&&&&&&
        \\
        &&\gate{\Lambda}&\qw&\qw&\multigate{1}{\ }&&
        &&\gate{\Lambda}&\gate{\theta}&\gate{\theta}&\qw&\multigate{1}{\ }&&
        \\
        &&& &\push{y_1}&\pureghost{\ }\cw&\push{\,b_1}\cw&
        &&&&&\push{y_2}&\pureghost{\ }\cw&\push{\,b_2}\cw&
        }
    \end{gather*}
    \caption{\label{fig:IDI0} An i.i.d.\ device with no public communication and no memory in channel-based DIQKD (left) and the introduction of the partial transpose (right). 
    We introduce this class of devices as $IDI_0$ in the Supplementary Material.
    }
\end{figure}
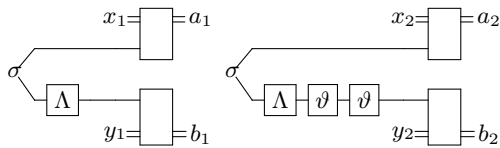

The above quantities give rise to quantities $\mathcal{P}_i^{\downarrow}$ which will be upper bounds on the actual device-independent privacy capacities. When combining them with an upper bound $ \mathcal{P}_i \leq \mathcal{E}$ we obtain (here illustrated with $i=2$):
\begin{align}
    \mathcal{P}_2^\text{DI} (\Lambda)  &
    \leq \mathcal{P}_{2}^{\downarrow}(\Lambda) 
    \leq \mathcal{E}    ^{\downarrow}(\Lambda).
\end{align}

Also here, we can now apply the partial transpose idea. In order to do so, we introduce the partial transpose map $\theta$ ($\theta(\rho) = \rho^ T$).
If a channel $\Lambda$ is such that $\theta\circ\Lambda$ is also a channel (i.e. $\Lambda$ completely positive and completely co-positive), then any device for $\Lambda$ can be transformed into a device for $\theta\circ\Lambda$ with the same statistics as shown in \Cref{fig:IDI0}. 
The consequence is analogous to \Cref{eq:pptupperbound} (for i=2): %
\begin{align}
    \mathcal{P}_2^\text{DI} (\Lambda) 
    \leq \mathcal{P}_{2}(\theta\circ\Lambda) 
    \leq \mathcal{E}    (\theta\circ\Lambda).
\end{align}
This bound can be used to show that there is a gap between the private capacity and the device-independent private capacity, as there exist examples of channels for which $\mathcal{P}_2(\Lambda)$ is large, but $\mathcal{P}_2(\theta\circ\Lambda)$ small~\cite{christandl2017relative}.

\paragraph*{Discussion.}
We have derived general upper bounds on the generation of device-independent key.
For the sake of completeness, we provide a detailed definition of device-independent key rates~\cite{ArnonFriedman2018,TowardsDIQKD} in the Supplemental Material.  
Using the upper bounds, we have shown that a gap can exist between the device-dependent (or standard) and device-independent distillable key. 
In fact, the gap can be shown to be maximally large, meaning that some states and channels support secret key generation, but at most a negligible amount of device-independent secret key. 
The construction has been obtained for a class of states and channels that have zero distillable entanglement or quantum capacity and that are known as PPT states or channels. 
We leave it as an interesting challenge to lower the dimension of such examples.

Notice that the partial transpose approach has previously led to upper bounds on Bell non-locality in terms of faithful measures of entanglement~\cite{HM15}, taking inspiration on upper bounds on key repeater rates~\cite{BCHW15}. 
In~\cite{ChristandlFerrara}, bounds on the key repeater rate were given beyond the use of the partial transpose idea, leading to a connection with distillable entanglement.
We hope that a similar result can be obtained connecting the device-independent distillable key and private capacity to the distillable entanglement and the quantum capacity, respectively~\cite{ChristandlFerrara}, potentially leading to bounds for non-PPT states and channels.

One may regard the gap $\Delta K(\rho,\mathcal{M})]\coloneqq K(\rho)-K^\text{DI}(\rho,\mathcal{M})$ as a measure of \emph{trust} towards a device $(\rho,\mathcal{M})$ (and analogously for quantum channels). 
For example, $\Delta K$ is zero for the singlet with CHSH testing, meaning that this device does not need to be trusted, and the same may hold for all pure states.
However, this is not the case for some bound-entangled states for which our results prove that $\Delta K$ is non-zero. 
Obtaining similar results for the multipartite case of conference key agreement is an interesting open problem.

\paragraph*{Note.}
After concluding the research on this article, we became aware of the independent but closely related work~\cite{AFL20}, where a conjecture is formulated that bound-entangled states have \emph{zero}
device-independent key against quantum adversary~\cite{AFL20} (see in this context the related results in case of non-signaling adversaries~\cite{WDH19}). Our work can be regarded as supporting evidence for this conjecture. 

\paragraph*{Acknowledgements.}
MC acknowledges financial support from the European Research Council (ERC Grant
Agreement No.~337603 and 81876) and VILLUM FONDEN via the QMATH Centre of Excellence (Grant No.~10059).
RF was supported by the Bundesministerium f\"ur Bildung und Forschung (BMBF) through Grant 16KIS0857 and thanks Jed Kaniewski for useful discussions.
KH acknowledges the support of the National Science Centre grant Sonata Bis 5 UMO-2015/18/E/ST2/00327, and partial support by the Foundation for Polish Science through the IRAP project, co-financed by the EU within the Smart Growth Operational Programme (contract no.~2018/MAB/5).


%

\section{Lower bounds on key}

In this Section, we derive and discuss in more detail the previous results and the states that we use to derive our lower bounds on the distillable key, used in the example in the main text. 

We parametrize Bell diagonal states as follows
\begin{equation*}
\rho_\mathrm{Bell} \coloneqq \frac{1}{2\alpha+2\beta}
\begin{bmatrix}
\alpha & 0 & 0 & \gamma \\
0 & \beta & \delta & 0 \\
0 & \delta & \beta & 0 \\
\gamma & 0 & 0 & \alpha 
\end{bmatrix}.
\end{equation*}
A lower bound on the distillable key obtained from measuring Bell diagonal states states was derived in~\cite[Eq.~(22)]{ABBBMM06} using Maurer's advantage distillation~\cite{Maurer93} and the Devetak-Winter~\cite{DW05} protocols.
These states were considered in \cite{ABBBMM06} where they where parametrized as $\lambda_1=\alpha+\gamma$, $\lambda_2=\alpha-\gamma$, $\lambda_3=\beta+\delta$, $\lambda_4=\beta-\delta$.
The values $\lambda_1,\dots,\lambda_4$ are the actual eigenvalues of the state and the probabilities of the Bell states.
The lower bound found in~\cite[Eq.~(22)]{ABBBMM06} is 
\begin{align*}
    K(\rho_\mathrm{Bell}^{\otimes m})
    &\geq
    1 -h(\varepsilon) - 
    \\
    &-(1-\varepsilon) h\left( \frac{1-\lambda_\mathrm{eq}^m}{2} \right)
        -\varepsilon  h\left( \frac{1-\lambda_\mathrm{dif}^m}{2} \right)
    \\
    &= 1 - h(\varepsilon)
    \\
    &-(1-\varepsilon) h\left( \frac{\alpha^m - \gamma^m}{2 \alpha^m} \right)
        -\varepsilon  h\left( \frac{\beta ^m - \delta^m}{2 \beta^m} \right)
    \yesnumber\label{eq:ABBBMM06.22}
    \\
    &\eqqcolon K_\mathrm{AD-DW}^m(\rho_\mathrm{Bell})
\end{align*}
where $h(x)=-x \log_2 x - (1-x)\log_2 (1-x)$ is the binary entropy and 
\[ \varepsilon =\frac{(\lambda_3+\lambda_4)^m}
  {(\lambda_1+\lambda_2)^m+(\lambda_3+\lambda_4)^m} = \frac{\beta^m}{\alpha^m+\beta^m}
\]
\begin{align*}
 \lambda_\mathrm{eq} &= \frac{|\lambda_1-\lambda_2|}{\lambda_1+\lambda_2} = \frac{\gamma}{\alpha},
 &
 \lambda_\mathrm{dif} &= \frac{|\lambda_3-\lambda_4|}{\lambda_3+\lambda_4} = \frac{\delta}{\beta}.
\end{align*}
For convenience we named the achieved rate above $K_\mathrm{AD-DW}^m(\rho_\mathrm{Bell})$, where AD-DW stands for Advantage-Distillation Devetak-Winter.
The bound can be simplified by recalling the following property of the entropy ($H((p_1, \dots, p_d)) = - \sum p_i \log_2 p_i$):
\begin{align*}
    &h(p) + p h(q) + (1-p)h(r) \\
    &= H((pq,p(1-q), (1-p)r, (1-p)(1-r) )
\end{align*}
for all $p,q,r\in[0,1]$ ($p$ is the probability of a control bit for another bit in either probability $q$ or $r$).
We thus have that \cref{eq:ABBBMM06.22} can be rewritten to 
\begin{align*}
&K_\mathrm{AD-DW}^m(\rho_\mathrm{Bell})  =
\nonumber\\
&  1 - H\left(\frac{(\alpha^m+\gamma^m,\alpha^m-\gamma^m,\beta^m+\delta^m,\beta^m-\delta^m)}{2\alpha^m+2\beta^m}\right).
\end{align*}
Note that the bound is invariant under multiplication by a constant of $\alpha$, $\beta$, $\gamma$ and $\delta$.
Also notice that this lower bound is the same as the the single-copy lower bound of %
\begin{equation}
\rho_{\mathrm{Bell},m} \coloneqq \frac{1}{2\alpha^m+2\beta^m}
\begin{bmatrix}
\alpha^m & 0 & 0 & \gamma^m \\
0 & \beta^m & \delta^m & 0 \\
0 & \delta^m & \beta^m & 0 \\
\gamma^m & 0 & 0 & \alpha^m 
\end{bmatrix}.
\label{eq:Bell-m}
\end{equation}
Namely, we have
\begin{equation}
    K_\mathrm{AD-DW}^m(\rho_\mathrm{Bell}) = K_\mathrm{AD-DW}^1(\rho_{\mathrm{Bell},m}).
    \label{eq:advantage-equivalence}
\end{equation}

\begin{corollary}
Consider the following $2d \times 2d$ ``block Bell diagonal'' state 
\begin{equation}
\rho = 
\begin{bmatrix}
A_1   & 0    & 0   & C \\
0     & B_1  & D   & 0 \\
0     &D^\dag& B_2 & 0 \\
C^\dag& 0    & 0   & A_2
\end{bmatrix},
\end{equation}
such that $\|A_1\|_1=\|A_2\|_1\eqqcolon \alpha$, $\|B_1\|_1=\|B_2\|_1 \eqqcolon \beta$, $2\alpha+2\beta=1$.
Let $\gamma\coloneqq\|C\|_1$, $\delta\coloneqq\|D\|_1$, then for all integers $m$ 
\begin{align}
&K(\rho^{\otimes m})  \geq
\nonumber\\
&  1 - H\left(\frac{(\alpha^m+\gamma^m,\alpha^m-\gamma^m,\beta^m+\delta^m,\beta^m-\delta^m)}{2\alpha^m+2\beta^m}\right).
\label{eq:lemma-Belllowerbound}
\end{align}
\label{thm:ps-ad-dw}
\end{corollary}

\begin{proof}
Consider the classical-classical-quantum (ccq) state obtained by computing the purification (the state of the eavesdropper) of $\rho$, tracing the qudit part, and then measuring the qubit part in the standard basis.
By the privacy-squeezing technique of~\cite{HHHO09}, the ccq state of $\rho$ is the same as the ccq state of
\begin{equation}
\rho_\mathrm{ps}=
\begin{bmatrix}
\|A_1\|_1 & 0 & 0 & \|C\|_1 \\
0 & \| B_1 \|_1& \|D\|_1 &0 \\
0& \|C^\dagger\|_1 & \|B_2 \|_1& 0\\
\|C^{\dagger}\|_1 & 0 & 0& \|A_2\|_1
\end{bmatrix}.
\end{equation}
Therefore any protocol that distills key only from the ccq state will produce the same amount of key for both states.
The advantage distillation protocol~\cite{Maurer93} and the Devetak-Winter protocol~\cite{DW05} are among these protocols and, in particular, they are the protocols used in~\cite{ABBBMM06} to obtain the lower bound \cref{eq:lemma-Belllowerbound} on the key of $\rho_\mathrm{ps}$.
Therefore \cref{eq:lemma-Belllowerbound} holds for $\rho$.
\end{proof}

Alternatively, we would like to bring attention to the fact that advantage distillation can be performed before privacy squeezing, leading to the same result. 
Indeed, as shown in~\cite{HHHO05}, using advantage distillation directly on the qubits of the states in \cref{thm:ps-ad-dw} results in 
\begin{equation}
\rho^m = 
\begin{bmatrix}
A_1^{\otimes m}   & 0    & 0   & C^{\otimes m} \\
0     & B_1^{\otimes m}  & D^{\otimes m}   & 0 \\
0     &D^{\dag\otimes m}& B_2^{\otimes m} & 0 \\
C^{\dag\otimes m}& 0    & 0   & A_2^{\otimes m}
\end{bmatrix}.
\end{equation}
Privacy-squeezing $\rho^m$ then results in $\rho_{\mathrm{Bell},m}$ from \cref{eq:Bell-m}.
Then \cref{thm:ps-ad-dw} follows from \cref{eq:advantage-equivalence}.
In this sense, privacy squeezing commutes with advantage distillation for these states.

\section{Upper bounds on partial-transpose key}
In this section, we derive and discuss in detail the upper bound on the distillable key of $\rho^\Gamma$ for a class of ``block Bell diagonal'' states (\cref{thm:ps-ad-dw}) that are PPT.
We use this upper bound in the example shown in the main text.

\begin{theorem}
Let $\rho_{ABA'B'} $ be a PPT block Bell diagonal state of the form 
\begin{equation}
\rho_{ABA'B'} = 
\begin{bmatrix}
\alpha A_1 & 0 & 0 & C \\
 0 & \beta B_1 & 0 & 0 \\
 0 & 0 & \beta B_2 & 0\\
C^\dag & 0 & 0 & \alpha A_2
\end{bmatrix},
\end{equation}
with $A_1$, $A_2$, $B_1$, $B_2$ separable states and $2\alpha + 2\beta =1$.
Then 
\begin{equation}
K(\rho^\Gamma) < 2 \beta,
\label{eq:lemma-PPTupperbound}
\end{equation}
where $\rho^\Gamma = \mathrm{id}\otimes T(\rho)$ is the partial-transposed state.
\label{thm:blockPPTupperbound}
\end{theorem}
\begin{proof}
We first observe that $\rho^{\Gamma}$ is a state, since $\rho \in PPT$. 
Next, $\rho^{\Gamma}$ can be expressed as  a convex combination of two states 
\[\rho^\Gamma = 2\alpha \rho'_\mathrm{corr} + 2 \beta \rho'_\mathrm{acorr}\] 
where 
\[\rho'_\mathrm{corr} = \frac{1}{2}\projector{00}\otimes A_1^{\Gamma} + \frac{1}{2}\projector{11}\otimes A_2^{\Gamma}\]
is a separable state, and 
\[\rho'_\mathrm{acorr} = \frac{1}{2}
\begin{bmatrix}
 0& 0 & 0 & 0 \\
 0 & B_1^{\Gamma} & C^\Gamma & 0 \\
 0 & C^{\dag\Gamma} & B_2^{\Gamma} & 0\\
 0 & 0 & 0 & 0
\end{bmatrix}
\]
is a state that becomes separable after dephasing either of the qubits.
In particular, the relative entropy of entanglement for these states reads
\[E_\mathrm{r}(\rho'_\mathrm{corr}) =0,\]
because the state is separable, and 
\[E_\mathrm{r}(\rho'_\mathrm{acorr}) \leq 1,\]
because the relative entropy is non-lockable~\cite{lock-ent} and dephasing a qubit can be done by applying a unitary, picked at random between the identity and the Pauli Z.
The non-lockability property then assures that $E_\mathrm{r}$ does not drop down under a von-Neumann measurement by more than the entropy of the random variable that samples the unitary transformation~\cite[Theorem~3]{HHHO09}.
Since the relative entropy after the random unitary is zero, it could not have been more than $h(\frac{1}{ 2})=1$ before it.
We thus have that 
\[K(\rho^{\Gamma})\leq E_\mathrm{r}(\rho^{\Gamma})\leq 2\beta E_\mathrm{r}(\rho_\mathrm{acorr}')\leq 2\beta,\]
where we used the facts that $K\leq E_\mathrm{r}$~\cite{HHHO05,HHHO09} and that $E_\mathrm{r}$ is convex.
\end{proof}

\section{Gap example}

\begin{figure}
	\includegraphics[width=8cm,trim={0cm 0cm 0cm 0cm}]{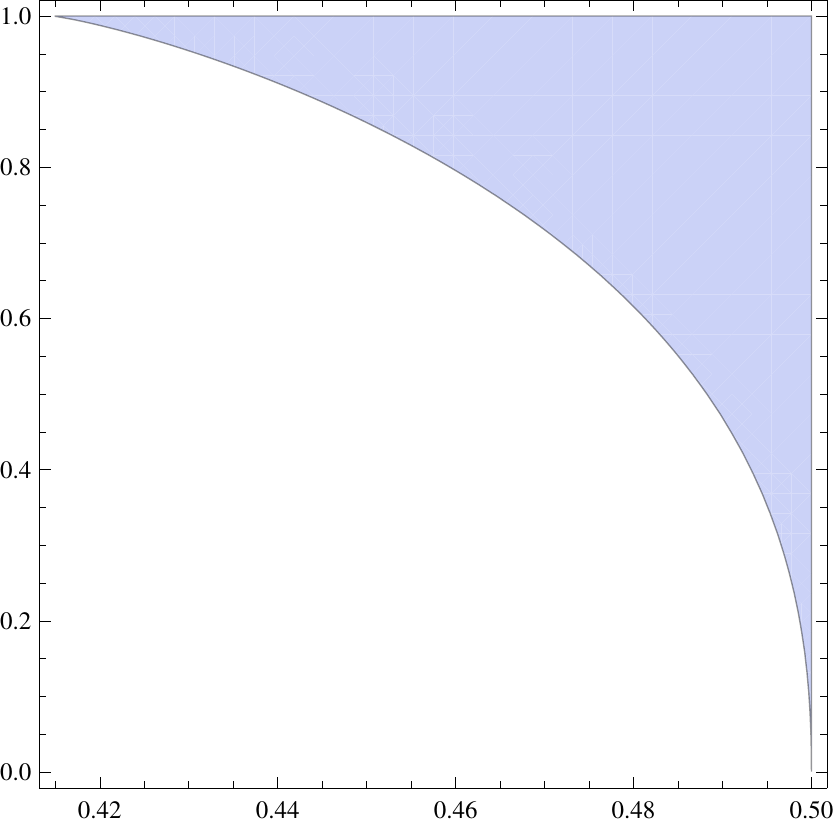}
	\caption{The shaded region is the set of pairs $(a,\alpha) \in [0,1]\times [0.415,0.5]$ leading to the gap between the device-independent key $K^\text{DI}$ and the device-dependent one $K$, for states of the form of \Cref{eq:ppt-gap-states}, according to the parametrization $\gamma =\alpha a $ and $\beta = \frac{1}{2} - \alpha$.
} \label {fig:region}
\end{figure}

Recall that from the main text it follows that a sufficient condition for a PPT state
to exhibit a gap $K^\mathrm{DI}(\rho) < K(\rho)$ is that $K(\rho^{\Gamma}) <K(\rho)$. 
We show below a sufficient condition for a wide class of states. 

\begin{theorem}
Let $\rho_{ABA'B'} $ be a PPT block Bell diagonal state of the form 
\begin{equation}
\rho_{ABA'B'} = 
\begin{bmatrix}
\alpha A_1 & 0 & 0 & C \\
 0 & \beta B_1 & 0 & 0 \\
 0 & 0 & \beta B_2 & 0\\
C^\dag & 0 & 0 & \alpha A_2
\end{bmatrix},
\label{eq:ppt-gap-states}
\end{equation}
with $A_1$, $A_2$, $B_1$, $B_2$ separable states and $2\alpha + 2\beta =1$.
Let $\gamma=\|C\|$, then the condition 
\begin{equation}
H(\alpha-\gamma,\alpha+\gamma,\beta,\beta) < 2 \alpha
\label{eq:separation}
\end{equation}
implies a gap $K^\text{DI}(\rho) < K_{D}(\rho)$.
\end{theorem}
\begin{proof}
The result directly follows from \cref{thm:ps-ad-dw} (for $m$= 1) and \cref{thm:blockPPTupperbound} and from what is said above.
The condition $1-H(\alpha-\gamma,\alpha+\gamma,\beta,\beta) > 2\beta$ is equivalent to \Cref{eq:separation} via $2\alpha +2\beta=1$. 
\end{proof}

The sufficient condition given in \Cref{eq:separation} can be further
expressed with $2$ parameters only, utilizing the normalization condition $2\alpha+2\beta=1$.
We therefore express $\gamma = \alpha a$ with $\alpha \in [0,1]$ and $\beta = (1 - 2\alpha)/2$, obtaining the equivalent condition as a function of $\alpha$ and $a$:
\begin{equation}
H\left((1+a)\alpha,(1-a)\alpha, \frac{1}{2} -\alpha,\frac{1}{2}-\alpha\right) < 2\alpha.
\end{equation}
The allowed region of parameters $(a,\alpha)$ that satisfy the above condition
is presented in \Cref{fig:region}.

\section{The State Distillable Key}

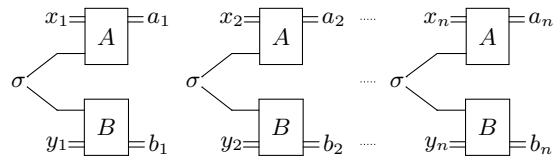
\begin{figure}
    \begin{equation*}
    \Qcircuit @C=.6em @R=.5em{
        &\push{x_1}&\puremultigate{1}{{A}}\cw &\push{\,a_1}\cw&
        &\push{x_2}&\puremultigate{1}{{A}}\cw &\push{\,a_2}\cw&
        &\dw&
        &\push{x_n}&\puremultigate{1}{{A}}\cw &\push{\,a_n}\cw&
        \\
        && \ghost{{A}} &&
        && \ghost{{A}} &&
        &&
        && \ghost{{A}} &&
        \\
        \push{\sigma\,} \ar@{-} [-1,1]\ar@{-} [1,1]&&&&
        \push{\sigma\,} \ar@{-} [-1,1]\ar@{-} [1,1]&&&&
        &\dw&
        \push{\sigma\,} \ar@{-} [-1,1]\ar@{-} [1,1]
        \\
        &&\multigate{1}{{B}} &&
        &&\multigate{1}{{B}} &&
        &&
        &&\multigate{1}{{B}} &&
        \\
        &\push{y_1}& \pureghost{{B}}\cw &\push{\,b_1}\cw&
        &\push{y_2}& \pureghost{{B}}\cw &\push{\,b_2}\cw&
        &\dw&
        &\push{y_n}& \pureghost{{B}}\cw &\push{\,b_n}\cw&
        }
    \end{equation*}
    \caption{\label{fig:state-devices}%
    The strongest restriction on the devices of the adversary in the case of state-based device-independent QKD. 
    This device is a special case of the most general device shown in \cref{fig:EAC}, which covers both state- and channels-based settings.
    Each round is an i.i.d.\ copy of the same measurements $\mathcal{N} = \{A^x_a, B^y_b\}$ on the same state $\sigma$.
    A cLOPC protocol around this device as in \cref{fig:DI-protocol} is a particular qLOPC protocol on $\sigma^{\otimes n}$.}
\end{figure}

The set of bipartite local measurements $\mathcal{M}$ of a device 
is a collection of POVMs for Alice and POVMs for Bob.
Together with a quantum state $\rho$, the result of measuring such POVMs will be a conditional probability distribution with two inputs and two outputs, creating equivalence classes of devices that generate the same distribution.
Recall that we denote this condition as $(\mathcal{M},\rho) \equiv (\mathcal{N},\sigma)$.
Additionally, an $\eps$ distance between conditional distributions will induce a $\eps$ distance between devices which we can denote with $(\mathcal{M},\rho) \approx_{\epsilon} (\mathcal{N},\sigma)$.
It is enough, for example, to consider the distance 
\[d(p,p') = \sup_{x,y}\|p(\cdot|x,y) - p'(\cdot|x,y)\|_1\leq \eps.\]

Informally, the device-independent distillable key of a state %
it is a supremum over all possible measurements $\mathcal{M}$ over the finite key rates $\kappa$ achieved by the best protocol on any device compatible with $(\mathcal{M},\rho)$, all in an appropriate asymptotic limit of blocklength and security/distance parameter.
This process is sufficiently general to include the recently proposed protocols of~\cite{ArnonFriedman2018} and realistic future protocols.
However,  for our purpose of upper bounding this quantity, considering all the compatible devices at every blocklength is exceedingly cumbersome.
To simplify the treatment, as mentioned in the main text, we define, as a relaxation, the larger device-independent distillable key $K^\text{DI}$, where we limit to compatible i.i.d.\ devices.
These are shown in \cref{fig:state-devices}.

Our definition follows the style of~\cite{karol-PhD,HHHO05,BCHW15}.
The key of a device is defined as:
\begin{align}
\label{eq:KDI}
K^\text{DI}(\mathcal{M},\rho) 
&\coloneqq \inf_{\eps >0} \limsup_{n\to\infty}  \kappa^{DI,\eps}_n(\mathcal{M},\rho)
\intertext{where}
\label{eq:KDI-eve}
\kappa^{DI,\eps}_{n}(\mathcal{M},\rho) 
&\coloneqq
\sup_{\substack{\Pi}} \inf _{\substack{(\mathcal{N},\sigma)\approx_\eps (\mathcal{M},\rho)}}
\kappa_n^\eps \qty( \Pi((\mathcal{N},\sigma)^{\otimes n}))
\end{align}
is the maximal key rate achieved for any security parameter $\eps$, 
blocklength or number of copies $n$, and measurement $\mathcal{M}$ 
chosen by Alice and Bob.

\begin{figure}
    \begin{equation*}
    \phantom{\sigma}
        \Qcircuit @C=1em @R=.75em{
            &&\qw                 &\qw&\multigate{1}{{A}^{\otimes n}}
            \\
            &&\puremultigate{2}{cLOPC}&\ustick{x^n}\cw&\pureghost{{A}^{\otimes n}}\cw&\ustick{a^n}\cw&\puremultigate{2}{cLOPC}\cw&\ustick{k}\cw
            \\
            \lstick{\sigma^{\otimes n}} \ar@{-} [-2,1] \ar@{-} [2,1]
            &&\pureghost{cLOPC}&      &                &             &\pureghost{LOPC}
            \\
            &&\pureghost{cLOPC}       &\ustick{y^n}\cw&\pureghost{{B}^{\otimes n}}\cw&\ustick{b^n}\cw&\pureghost{cLOPC}\cw&\ustick{\tilde{k}}\cw
            \\
            &&\qw                 &\qw&\multigate{-1}{{B}^{\otimes n}}
        }
    \end{equation*}
    \caption{%
    A device-independent protocol will use cLOPC to generate the classical inputs $x^n$ and $y^n$ to $n$ copies of the device $\qty(\mathcal{N}=\{A^x_a,B^y_b\},\sigma)$, and process the outputs $a^n$ and $b^n$.
    The composition of the measurements and the cLOPC protocols is a qLOPC protocol acting on the state $\sigma$ as a device-dependent protocol.}
    \label{fig:DI-protocol}
\end{figure}
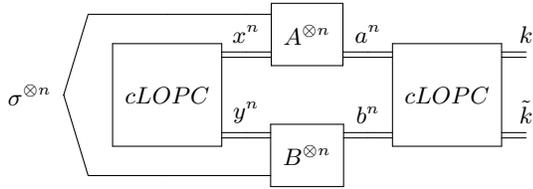

$\Pi$ is a protocol composed of classical local operations and public communication (cLOPC) acting on $n$ identical copies of the device $(\mathcal{N},\sigma)$, which composed with the measurement, results in a protocol of quantum local operations and public communication (qLOPC), as displayed in \cref{fig:DI-protocol}.
$\kappa_n^\eps(\Pi,(\mathcal{N},\sigma))$ is the amount of $\eps$-perfect key rate achieved at the output;
since our result does not depend on its expression which also varies with the security criteria, for an explicit definition of $\kappa_n^\eps$ we refer to~\cite{karol-PhD,HHHO05,BCHW15,ArnonFriedman2018} and references therein.
The key of a state is then the maximum over the choice of measurements
\begin{equation}
    K^\text{DI}(\rho) \coloneqq \sup_{\mathcal{M}} K^\text{DI}(\mathcal{M},\rho).
\end{equation}

$K^\text{DI}(\rho)$ is the largest operationally justified definition of the device-independent key.
This definition assumes that Alice and Bob have determined the best measurement for the state $\rho$ and allows the attacker to optimize over all the classical information.
At the same time, $K^\text{DI}(\rho)$ restricts the attacker to i.i.d.~device attacks, corresponding to the so-called collective attack in the case of device-dependent quantum key distribution. 
Although there is currently no analog of the quantum de Finetti theorem for device-independent QKD~\cite{Rotem-phd,Christandl2007one,renner2005security}, it is possible that general attacks and i.i.d. attacks are also equally powerful in the device-independent scenario. \cite{ArnonFriedman2018} shows that this is the case for the CHSH game.

To obtain our bounds, we relax into non-operational rates. 
We further weaken the adversary by taking away his knowledge about the public information, which is possible because the cLOPC protocol with the measurement forms a particular qLOPC protocol.
Namely, we have that the finite rate of \cref{eq:KDI-eve} satisfies for all devices  $(\mathcal{N},\sigma)$ and cLOPC protocols $\Pi$
\begin{equation}
\label{eq:cLOPC-qLOPC}
\kappa_n^\eps \qty( \Pi((\mathcal{N},\sigma)^{\otimes n}))
\leq \sup_{\Pi'\in \text{qLOPC}} \kappa_n^\eps \qty( \Pi'(\sigma^{\otimes n}))
\end{equation}
where $\Pi'$ in the right-hand side are now qLOPC protocols independent of the measurement $\mathcal{N}$.
Notice that 
\begin{equation}
    \label{eq:KD-boxed}
    \kappa_{n}^\eps(\sigma) \coloneqq \sup_{\Pi'\in \text{qLOPC}} \kappa_n^\eps \qty( \Pi'(\sigma^{\otimes n}))
\end{equation}
is the best device-dependent finite rate which asymptotically leads to the device-dependent distillable key
\begin{align}
\label{eq:KD}
K(\rho) 
&\coloneqq \inf_{\eps >0} \limsup_{n\to\infty} \kappa_{n}^\eps(\rho)
\end{align}
in the style of~\cite{karol-PhD,HHHO05,BCHW15}.
By taking infimum over devices $(\mathcal{N},\sigma)$ on both sides of \cref{eq:cLOPC-qLOPC} we obtain
\begin{align*}
\kappa^{DI,\eps}_{n}(\mathcal{M},\rho) 
& \leq
\inf_{(\mathcal{N},\sigma)\approx_\eps (\mathcal{M},\rho)} \kappa_{n}^\eps(\sigma).
\end{align*}
Taking further the infimum over $\epsilon$ and limit of large $n$ in, we obtain
\begin{align}
K^\text{DI}(\mathcal{M},\rho)
& \leq
\inf_{\eps >0} \limsup_{n\to\infty} 
\inf_{(\mathcal{N},\sigma)\approx_\eps (\mathcal{M},\rho)}
\kappa_n^\eps(\sigma).
\label{eq:KDvsKDI}
\end{align}

We summarize the proof of the above inequality in \cref{eq:KDI-start} below.
The main result which follows from the above inequality (see \cref{thm:main_appendix} below) reads
\begin{equation}
\label{thm:main}
K^\text{DI}(\rho)\leq 
{K}^{\downarrow} (\rho)\coloneqq 
\sup_\mathcal{M} \inf_{(\mathcal{N},\sigma)\equiv (\mathcal{M},\rho)} K(\sigma).
\end{equation}
Before we proceed with the proof of the above
statement we prove \cref{eq:KDvsKDI}. 

\begin{lemma}
\label{eq:KDI-start}
\begin{align}
K^\text{DI}(\mathcal{M},\rho)
& \leq
\inf_{\eps >0} \limsup_{n\to\infty} 
\inf_{(\mathcal{N},\sigma)\approx_\eps (\mathcal{M},\rho)}
\kappa_n^\eps(\sigma)
\end{align}
where $\kappa_n^\eps(\sigma)$ is the finite parameter device-dependent rate of \cref{eq:KD-boxed}.
\end{lemma}
\begin{proof}
By simple max-min inequality we can swap the order of the optimization to get an upper bound,
and by then relaxing to all device-dependent protocols we have 
{\allowdisplaybreaks
\begin{align}
\nonumber
\kappa^{DI,\eps}_{n}&(\mathcal{M},\rho) 
\\
&\coloneqq
\sup_{\substack{\Pi\in\\\text{cLOPC}}} \inf _{(\mathcal{N},\sigma)\approx_\eps (\mathcal{M},\rho)}
\kappa_n^\eps\qty(\Pi((\mathcal{N},\sigma)^{\otimes n}))
\\
\nonumber
&\leq \inf _{(\mathcal{N},\sigma) \approx_\eps (\mathcal{M},\rho)}
\sup_{\substack{\Pi\in\\\text{cLOPC}}} \kappa_n^\eps\qty(\Pi((\mathcal{N},\sigma)^{\otimes n}))
\\
\label{eq:lopc-locc}
&\leq \inf _{(\mathcal{N},\sigma) \approx_\eps (\mathcal{M},\rho)}
\sup_{\substack{\Pi'\in\\\text{qLOPC}}} \kappa_n^\eps\qty(\Pi'(\sigma^{\otimes n}))
\\
\label{eq:kappa-dd}
&= \inf _{(\mathcal{N},\sigma) \approx_\eps (\mathcal{M},\rho)} 
\kappa_n^\eps(\sigma)
\end{align}}
where the rates in \Cref{eq:lopc-locc,eq:kappa-dd}
are the same rates introduced in \Cref{eq:KD,eq:KD-boxed}.
The relaxation of the protocols in \Cref{eq:lopc-locc} is clear and displayed in \Cref{fig:DI-protocol};
the measurement $\mathcal{N}$ acts like a fixed pre-processing on the state,
and thus for any protocol $\Pi$ acting on the devices,
the composition of $\Pi$ with the $n$ measurements $\mathcal{N}$
is just a particular protocol acting on $n$-copies of the state $\sigma$.
Removing this constraint can only increase the rate.
Plugging in the definition at \cref{eq:KDI}, we thus have
\begin{align}
K^\text{DI}(\mathcal{M},\rho)
& \leq
\inf_{\eps >0} \limsup_{n\to\infty} 
\inf_{(\mathcal{N},\sigma)\approx_\eps (\mathcal{M},\rho)}
\kappa_n^\eps(\sigma)
\end{align}
and the proof is concluded.
\end{proof}

Before proving the main result of this section
we need a technical observation.

\begin{observation}
    \label{lemma:maxmin}
    The max-min inequality is also valid as a limsup-inf inequality.
    Namely, for any sequence of functions $f_n(x)$, we have 
    \begin{align}
        \label{eq:limsup-inf}
        \limsup_{n\to\infty} \inf_x f_n(x)
        \leq \inf_x \limsup_{n\to\infty} f_n(x)
    \end{align}    
    Indeed, we can rewrite the limit superior using infimum and supremum,
    and then use max-min inequality followed by the commutation of two infima:
    \begin{align*}
        \limsup_{n\to\infty} \inf_x f_n(x)
        &=    \inf_{ n\geq 0} \sup_{m\geq n} \inf_x f_n(x)
        \\
        &\leq \inf_{ n\geq 0} \inf_x \sup_{m\geq n} f_n(x)
        \\
        &= \inf_x \limsup_{n\to\infty} f_n(x).
    \end{align*}

\end{observation}
We are ready to prove the main result.

\begin{theorem} For any bipartite state $\rho$ it holds
\begin{equation}
K^\text{DI}(\rho)\leq 
{K}^{\downarrow} (\rho)\coloneqq 
\sup_\mathcal{M} \inf_{(\mathcal{N},\sigma)\equiv (\mathcal{M},\rho)} K(\sigma).
\end{equation}
\label{thm:main_appendix}
\end{theorem}

\begin{proof}
    We use \Cref{eq:KDI-start} as our starting point
	and use \cref{lemma:maxmin}
	\begin{align}
        K^\text{DI}(\mathcal{M},\rho)
        & \leq
        \inf_{\eps >0} 
        \limsup_{n\to\infty} 
        \inf_{(\mathcal{N},\sigma)\approx_\eps (\mathcal{M},\rho)}
        \kappa_n^\eps(\sigma)
        \\
        & \leq
        \inf_{\eps >0} 
        \inf_{(\mathcal{N},\sigma)\approx_\eps (\mathcal{M},\rho)}
        \limsup_{n\to\infty} 
        \kappa_n^\eps(\sigma).
    \intertext{We can always restrict the infimum
    to devices that are exactly equal to the original box, 
    this only reduces the set of devices}
        K^\text{DI}(\mathcal{M},\rho)
        & \leq
        \inf_{\eps >0} 
        \inf_{(\mathcal{N},\sigma)\equiv (\mathcal{M},\rho)}
        \limsup_{n\to\infty} 
        \kappa_n^\eps(\sigma).
    \intertext{Since the infimum over devices is now independent of the security parameter,
    we can now simply commute the two infima}
        K^\text{DI}(\mathcal{M},\rho)
        & \leq
        \inf_{(\mathcal{N},\sigma)\equiv (\mathcal{M},\rho)}
        \inf_{\eps >0}
        \limsup_{n\to\infty} 
        \kappa_n^\eps(\sigma)
        \\
        &=
        \inf_{(\mathcal{N},\sigma)\equiv (\mathcal{M},\rho)}
        K(\sigma)
	\end{align}
	reaching the claim.
\end{proof}

\section{Channel Private Capacity with two-way, one-way or no classical communication}

The same idea presented in this paper for states also works for the private (or secret) capacity $\mathcal{P}(\Lambda)$ of a channel $\Lambda$, and thus for the most general setting for QKD, which includes, for example, modelling the optical fiber itself instead of the states produced across the optical fiber. 

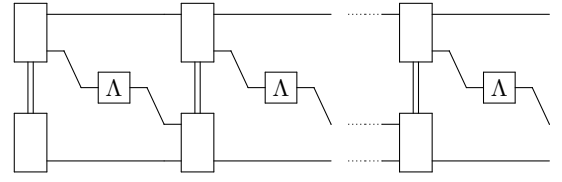
\begin{figure}
    \begin{gather*}
    \Qcircuit @C=.7em @R=.5em @!R=1em {
        \puremultigate{1}{\ } &&&&&\qw[-5]&
            \multigate{1}{\ } &&&&&\qw[-5]&
        &\dw&\dw& 
            \multigate{1}{\ } &&&&&\qw[-5]&
        \\
        \pureghost{\ }&\ar@{-}[1,1]\qw&&&&&
        \pureghost{\ }&\ar@{-}[1,1]\qw&&&&&
        &&&
        \pureghost{\ }&\ar@{-}[1,1]\qw&&&&&
        \\
        &&&\gate{\Lambda}&\ar@{-}[1,1]\qw&&
        &&&\gate{\Lambda}&\ar@{-}[1,1]\qw&&
        &&&
        &&&\gate{\Lambda}&\ar@{-}[1,1]\qw&&
        \\
        \puremultigate{1}{\ }\cwx[-2]&&&&&&
            \multigate{1}{\ }\cwx[-2]&&&&&&
        &\dw&\dw&
            \multigate{1}{\ }\cwx[-2]&&&&&&
        \\
        \pureghost{\ }&&&&&\qw[-5]&
            \ghost{\ }&&&&&\qw[-5]&
        &\dw&\dw&
            \ghost{\ }&&&&&\qw[-5]&
        }
    \end{gather*}
    \caption{\label{fig:DD-protocol}%
    A general LOPC protocol for channel $\Lambda$. 
    At the beginning and after each channel use, Alice and Bob are allowed to perform an LOPC operation to prepare the input for the next channel use. 
    In device-independent QKD, the quantum parts of the LOPC protocols are hidden behind the devices, and only classical protocols can occur outside the devices.}
\end{figure}

We mentioned that the private capacity has different versions, namely the two-way  ($\mathcal{P}_2$), one-way ($\mathcal{P}_1$), or direct ($\mathcal{P}_0$) private capacities depending on whether two-way  ($LOPC_2$) or one-way classical communication ($LOPC_1$), or only local operations (``0-way'' communication, $LOPC_0 = LO$) are allowed in the privacy protocol (practical protocols might still need communication for practical purposes, like testing, outside/around the privacy protocol). With increasing power comes increasing rates and thus $\mathcal{P}_0 \le \mathcal{P}_1 \le \mathcal{P}_2$.
These device-dependent private capacities are all defined as 
\begin{align}
\label{eq:P}
\mathcal{P}_i(\Lambda) 
&\coloneqq 
\inf_{\eps >0} \limsup_{n\to\infty} \pi^\eps_{i,n}(\Lambda),
\end{align}
where $\pi_{i,n}^\eps(\Lambda)$ is the largest $\eps$-perfect key rate obtained by the best privacy protocol (with $i$-way communication) that uses $n$ identical copies of $\Lambda$.
A general protocol around $n$ i.i.d.\ copies of $\Lambda$ is displayed in \Cref{fig:DD-protocol}.

In the device-independent channel setting, given a channel $\Lambda$ from Alice to Bob, we define an honest device for $\Lambda$ as a tuple $(\mathcal{M}, \rho, \Lambda)$, where $\rho$ is a bipartite state of Alice and the input to the channel, and $\mathcal{M}$ is a device measurement of Alice and Bob (the output of the channel).
The conditional probability distribution is then obtained via
\begin{equation*}
    p_{(\mathcal{M}, \rho, \Lambda)}(ab|xy) = \tr[(\id \otimes \Lambda)(\rho) \cdot M_a^x\otimes \tilde M_b^y],
\end{equation*}
and we have the same definitions of equality and distance for two devices.
Again we define 
\begin{align*}
    (\mathcal{M}, \rho, \Lambda)&\equiv (\mathcal{M}', \rho', \Lambda'),
    \\
    (\mathcal{M}, \rho, \Lambda)&\approx_\eps (\mathcal{M}', \rho', \Lambda')
\intertext
{as the conditions}
p_{(\mathcal{M}, \rho, \Lambda)}&= p_{(\mathcal{M}', \rho', \Lambda')},
\\
p_{(\mathcal{M}, \rho, \Lambda)}&\approx_\eps p_{(\mathcal{M}', \rho', \Lambda')}.
\end{align*}

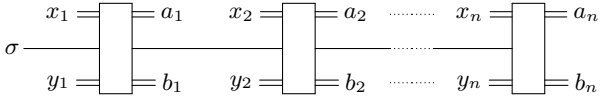
\begin{figure}
    \begin{equation*}
        \Qcircuit @C=.9em @R=.5em{
            &\push{x_1\;}&\puremultigate{2}{\ }\cw &\push{\;a_1}\cw&
            &\push{x_2\;}&\puremultigate{2}{\ }\cw &\push{\;a_2}\cw&
            &\dw&\dw
            &\push{x_n\;}&\puremultigate{2}{\ }\cw &\push{\;a_n}\cw&
            \\
            \push{\sigma\,} 
            &\qw & \ghost{\ } &\qw&\qw
            &\qw & \ghost{\ } &\qw&\qw
            &\dw&\dw
            &\qw & \ghost{\ } %
            \\
            &\push{y_1\;}& \pureghost{\ }\cw &\push{\;b_1}\cw&
            &\push{y_2\;}& \pureghost{\ }\cw &\push{\;b_2}\cw&
            &\dw&\dw
            &\push{y_n\;}& \pureghost{\ }\cw &\push{\;b_n}\cw&
        }
    \end{equation*}
    \caption{\label{fig:EAC}This device is the most general way an adversary could implement any device in DIQKD, whether the honest implementation is i.i.d., state-based, channel-based, or neither.
    Single lines are quantum systems; double lines are classical systems; $x_i,y_i$ are the inputs, and $a_i,b_i$ the outputs of the device.
    This device is almost the most general case allowed by the entropy-accumulation theorem (EAT)~\cite{ArnonFriedman2018}.
    Indeed, an entropy-accumulation channel (EAC)~\cite{ArnonFriedman2018} is obtained by joining such a device with the inputs generated by a Markov chain and copying the inputs as additional outputs.}
\end{figure}

At this point, recall that we can distinguish between the classical communication used by Alice and Bob outside the device in the classical distillation protocol and the communication used by Eve inside the device to produce the quantum states to be measured.
We will thus define various classes of devices.
Notice that all the classes of devices that we will define are a special case of the devices in \cref{fig:EAC} where the adversary is allowed to do anything in between the rounds and is only required to provide two pairs of classical input-outputs.
We remark on this because the devices of \cref{fig:EAC} are part of the so-called entropy-accumulation channels for which the best information-theoretic tool currently available obtains the largest achieved device-independent key rates~\cite{ArnonFriedman2018}.

\begin{figure}
    \begin{gather*}
    \Qcircuit @C=.7em @R=.5em @!R=1em {
        &&&&&&\push{x_1\;}&\puremultigate{1}{\ }\cw&\push{\;a_1}\cw
        &&&&&&\push{x_2\;}&\puremultigate{1}{\ }\cw&\push{\;a_2}\cw
        \\
        \puremultigate{1}{\ } &&&&&\qw[-5]&\multigate{1}{\ }&\ghost{\ }&
            \multigate{1}{\ } &&&&&\qw[-5]&\multigate{1}{\ }&\ghost{\ }&
        \qw&\dw&\dw
        \\
        \pureghost{\ }&\ar@{-}[1,1]\qw&&&&&\pureghost{\ }&&
        \pureghost{\ }&\ar@{-}[1,1]\qw&&&&&\pureghost{\ }&&
        \pureghost{\ }
        \\
        &&&\gate{\Lambda}&\ar@{-}[1,1]\qw&&&&
        &&&\gate{\Lambda}&\ar@{-}[1,1]\qw&&&&
        \\
        \puremultigate{1}{\ }\cwx[-2]&&&&&&\multigate{1}{\ }\cwx[-2]&&
        \puremultigate{1}{\ }\cwx[-2]&&&&&&\multigate{1}{\ }\cwx[-2]&&
        \\
        \pureghost{\ }&&&&&\qw[-5]&\ghost{\ }&\multigate{1}{\ }&
            \ghost{\ }&&&&&\qw[-5]&\ghost{\ }&\multigate{1}{\ }&
        \qw&\dw&\dw
        \\
        &&&&&&\push{y_1\;}&\pureghost{\ }\cw&\push{\;b_1}\cw
        &&&&&&\push{y_2\;}&\pureghost{\ }\cw&\push{\;b_2}\cw
        }
    \end{gather*}
    \caption{\label{fig:DI2}%
    A $DI_2$ device, a channel device on i.i.d.\ copies of the channel $\Lambda$ that uses two-way public communication and memory between the rounds to generate the input state for the channel uses.
    To obtain the other $DI_j$ devices, restrict the classical communication (double lines).
    To get the $IDI_j$ devices, remove the memory (single lines) connecting the input/output rounds to the next channel input.
    The cLOPC protocols used by Alice and Bob to distill the key connect to the classical lines in this diagram, i.e., by connecting the cLOPC protocol to the input-output rounds as done in \cref{fig:DI-protocol}}
\end{figure}
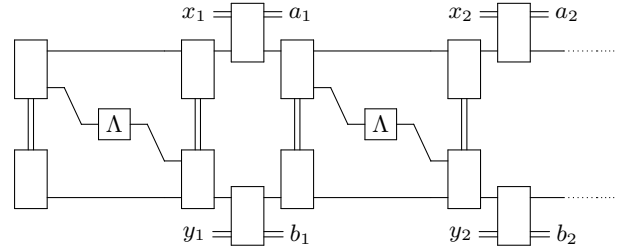

\begin{figure}
    \begin{equation*}
    \Qcircuit @C=.7em @R=.5em @!R=1em {
        &&&&&&\push{x_1\;}&\puremultigate{1}{\ }\cw&\push{\;a_1}\cw
        &&&&&&\push{x_2\;}&\puremultigate{1}{\ }\cw&\push{\;a_2}\cw
        \\
        \puremultigate{1}{\ } &&&&&\qw[-5]&\qw&\ghost{\ }&
            \multigate{1}{\ } &&&&&\qw[-5]&\qw&\ghost{\ }&
        \qw&\dw&\dw
        \\
        \pureghost{\ }&\ar@{-}[1,1]\qw&&&&&\pureghost{\ }&&
        \pureghost{\ }&\ar@{-}[1,1]\qw&&&&&\pureghost{\ }&&
        \pureghost{\ }
        \\
        &&&\gate{\Lambda}&\ar@{-}[1,1]\qw&&&&
        &&&\gate{\Lambda}&\ar@{-}[1,1]\qw&&&&
        \\
        &&&&&&\multigate{1}{\ }&&
        &&&&&&\multigate{1}{\ }&&
        \\
        &&&&&   &\pureghost{\ }&\multigate{1}{\ }&
        &&&&&\qw[-6]&\ghost{\ }&\multigate{1}{\ }&
        \qw&\dw&\dw
        \\
        &&&&&&\push{y_1\;}&\pureghost{\ }\cw&\push{\;b_1}\cw
        &&&&&&\push{y_2\;}&\pureghost{\ }\cw&\push{\;b_2}\cw
        }
    \end{equation*}
    \caption{\label{fig:DI0}%
    A $DI_0$ device, a restriction of $DI_2$ devices where the adversary is not allowed to use classical communication between the two sides of the device in the preparation of the states with the i.i.d.\ channel $\Lambda$. 
    Here the input state is generated at each round from a quantum memory, while in the main text, we are also restricted to single-copy i.i.d.\ states ($IDI_0$ devices).
    Both lead to the same bound.
    A cLOPC protocol around this device is still a particular qLOPC protocol around $\Lambda^{\otimes n}$. }
\end{figure}
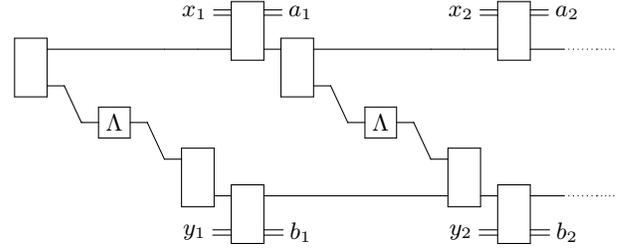

We denote $DI_0$, $DI_1$ and $DI_2$ the devices where the channel is i.i.d., memory is allowed, and respectively use none, one-way or two-way public communication between the input-output rounds. 
Notice that this communication does not happen between Alice and Bob giving their classical inputs and receiving their classical outputs (which would not allow for device independence), but either before the inputs are given or after the outputs are obtained.
The largest of the classes, $DI_2$, is displayed in \cref{fig:DI2}.
The $DI_j$ devices can still share memory locally at Alice and Bob across each round.
Thus we can further restrict the adversary as mentioned above and define the (i.i.d.-device independent) variants $IDI_0$, $IDI_1$, and $IDI_2$, where the devices are i.i.d.\ and are not allowed memory or communication from one round to the next. 
Notice that the i.i.d.\ assumption in the case of channels is much stronger and is unnatural compared to the state case because even if the channel itself in the device is i.i.d., the device might be not-i.i.d.\ because of the input state.
In general, even Alice and Bob need to use non-i.i.d input states to achieve the private capacity~\cite{SRS08}.
In contrast, in the state case, the de Finetti reduction shows that we can assume i.i.d input states without loss of generality~\cite{Renner07,CKR09}.
Therefore our use of $IDI_0$ devices is purely technical.
\Cref{fig:DI0} displays the class of devices $DI_0$, and cutting the lines connecting the input-output measurement with the next round produces $IDI_0$.

We can now define the corresponding device-independent private capacity for each choice of adversarial devices and each choice of available communication to Alice and Bob.
In order to do that, we want to change $\pi_{i,n}^\eps$, as for the state scenario, and define for $i,j= 0,1,2$
\begin{align}
\label{eq:PDI}
\mathcal{P}_i^{DI_j}(\mathcal{M},\rho,\Lambda) 
&\coloneqq 
\inf_{\eps >0} \limsup_{n\to\infty} \pi^{DI_j,\eps}_{i,n}(\mathcal{M},\rho,\Lambda),
\end{align}
where $\pi^{DI_j,\eps}_{i,n}$ will be the largest key rate optimized over privacy protocols, this time also including a minimization over the possible devices in $DI_j$ that are compatible with the honest device. 
Again, the upper bounds that we are interested in are upper bounds for all these capacities, and thus for simplicity, it will suffice to focus on just the i.i.d.\ devices $IDI_j$.
This will allow us to define just $\pi^{IDI_j,\eps}_{i,n}$, which is less cumbersome than defining $\pi^{DI_j,\eps}_{i,n}$.
The largest of these capacities is $\mathcal{P}_2^{IDI_0}$, since i.i.d.\ devices, larger $i$ and smaller $j$ make for larger rates. 
Private capacities with $j<i$ are arguably less meaningful because they allow less classical communication between the devices than is permitted to Alice and Bob.
Therefore, they should be considered for completeness and more as mathematical tools.

Each device-independent private capacity $\mathcal{P}_i^{DI_j}$ or $\mathcal{P}_i^{IDI_j}$  is upper bounded by a device-dependent capacity, as a consequence of defining and upper bounding the corresponding finite rates as follows.
Each choice of $i$ and $j$ defines the finite key rates rates $\pi^{DI_j,\eps}_{i,n}$ and $\pi^{DI_j,\eps}_{i,n}$, both bounded by a device-dependent finite key rate, by making the device-independent protocol, together with the state and measurement of the device, a specific device-dependent protocol: 
\begin{align}
\nonumber
\pi&{}^{DI_j,\eps}_{i,n}(\mathcal{M},\rho,\Lambda) 
\\
&\leq \pi^{IDI_j,\eps}_{i,n}(\mathcal{M},\rho,\Lambda) 
\\
&\coloneqq
\sup_{\Pi\in cLOPC_i} 
\inf _{\substack{
    (\mathcal{N},\sigma, \Lambda') \in IDI_j
    \\ 
    (\mathcal{N},\sigma, \Lambda' )\approx_\eps (\mathcal{M},\rho, \Lambda) 
    }}
\kappa^\eps_n(\Pi,(\mathcal{N},\sigma,\Lambda'))
\allowdisplaybreaks
\\
&\leq
\inf _{\substack{
    (\mathcal{N},\sigma, \Lambda') \in IDI_j
    \\ 
    (\mathcal{N},\sigma, \Lambda' )\approx_\eps (\mathcal{M},\rho, \Lambda) 
    }}
\sup_{\Pi\in cLOPC_i} 
\kappa^\eps_n(\Pi,(\mathcal{N},\sigma,\Lambda'))
\\
&\leq
\inf _{\substack{
    (\mathcal{N},\sigma, \Lambda') \in IDI_j
    \\ 
    (\mathcal{N},\sigma, \Lambda' )\approx_\eps (\mathcal{M},\rho, \Lambda) 
    }}
\sup_{\Pi' \in qLOPC_{\max\{i,j\}}} 
\kappa^\eps_n(\Pi',\Lambda')
\label{eq:maxij}
\\
&=
\inf _{\substack{
    (\mathcal{N},\sigma, \Lambda') \in IDI_j
    \\ 
    (\mathcal{N},\sigma, \Lambda' )\approx_\eps (\mathcal{M},\rho, \Lambda) 
    }}
\pi^\eps_{\max\{i,j\},n}(\Lambda')
\\
&\leq
\inf _{\substack{
    (\mathcal{N},\sigma, \Lambda') \in IDI_j
    \\ 
    (\mathcal{N},\sigma, \Lambda' )\approx_\eps (\mathcal{M},\rho, \Lambda) 
    }}
\pi^\eps_{2,n}(\Lambda')
\end{align}
where $\kappa^\eps_n$ is the rate of achieved $\eps$-perfect key, while $\pi^\eps_{\ell, n}$ is the device-dependent finite rate already optimized over $\ell$-way protocols acting on $n$ copies of $\Lambda'$.
Taking the limits, and with the same arguments as \Cref{thm:main}, gives upper bounds of the device-independent capacities $\mathcal{P}_i^{DI_j}$ in terms of optimized device-dependent capacities $\mathcal{P}_{\max\{i,j\}}$:
\begin{align}
    \mathcal{P}_i^{DI_j} (\mathcal{M},\rho, \Lambda) \leq 
    \inf _{\substack{
        (\mathcal{N},\sigma, \Lambda') \in IDI_j
        \\ 
        (\mathcal{N},\sigma, \Lambda' )\eqq(\mathcal{M},\rho, \Lambda) 
        }} 
    \mathcal{P}_{\max\{i,j\}}(\Lambda'),
\end{align}
and in particular, with $\ell \coloneqq\max\{i,j\}$,
\begin{align}
    \mathcal{P}_i^{DI_j} (\Lambda) 
    &\coloneqq \sup_{\mathcal{M},\rho} \mathcal{P}_i^{DI_j} (\mathcal{M},\rho, \Lambda)
    \\
    \leq \mathcal{P}_{\ell}^{\downarrow_j}(\Lambda) 
    &\coloneqq 
    \sup_{\mathcal{M},\rho} 
    \inf _{\substack{
        (\mathcal{N},\sigma, \Lambda') \in IDI_j
        \\ 
        (\mathcal{N},\sigma, \Lambda' )\eqq(\mathcal{M},\rho, \Lambda) 
        }} 
    \mathcal{P}_{\ell}(\Lambda')
    \\
    \leq \mathcal{P}_{2}^{\downarrow_0}(\Lambda) 
\end{align}

\begin{figure}
    \centering
    \newcommand{\vgeq}{\rotatebox{90}{$\leq$}}
    \begin{tikzpicture}[xscale=2,yscale=-1]
    \foreach \i in {0,1,2} {
        \foreach \j in {0,1,2} {
            \node (p\i\j) at (\i,\j) {$\mathcal{P}_{\i}^ {IDI_\j}$};
        }
    }

    \foreach \j/\k in {0/1,1/2} {
        \node (P2\k) at (3 ,\k) {$\mathcal{P}_{ 2}^ {\downarrow_\k}$};
        \node (P\j0) at (\j,-1) {$\mathcal{P}_{\j}^ {\downarrow_0 }$};

        \node at (\j.5,0) {$\leq$};
        \node at (2,\j.5) {\vgeq};
        \foreach \i in {0,1,2} {
            \node at (\i.5,\k) {$\leq$};
            \node at (\j,\i-.5) {\vgeq};
        }
    }

    \node (P20) at (3,-1) {$\mathcal{P}_{2}^ {\downarrow_0}$};
    \node at (3,1.5) {\vgeq};
    \node at (.5,-1) {$\leq$};
    
    \node at (3,-.3) {\vgeq};
    \node at (2.3,-1) {$\leq$};
    \node at (2.5,-.5) {\rotatebox{45}{$\leq$}};

    \end{tikzpicture}
    \caption{Relationship between the device-independent channel capacities.}
    \label{fig:channel-capacities}
\end{figure}

\begin{remark*}
Notice that we did not prove that $\mathcal{P}_i^{DI_j} \leq \mathcal{P}_{i}^{\downarrow_j}$, the crucial step being \eqref{eq:maxij}.
Take the example of $i=1$ and $j=2$: Alice and Bob are allowed only one-way communication, whereas the devices could use two-way communication.
The possibility of $\mathcal{P}_1^{DI_2} > \mathcal{P}_{1}^{\downarrow_2}$ means that the added power of two-way communication could allow the device to switch to a channel $\Lambda'$ with bad one-way private capacity, e.g., $\mathcal{P}_1(\Lambda')=0$.
However, to mimic the statistics of the original channel, some key needs to be extracted using the two-way communication increasing $\mathcal{P}_1^{DI_2}$ but not $\mathcal{P}_1^{\downarrow_2}$.
Still, this could be a better attack than simply finding the worst replacement using only one-way communication.

In other words, we could have the following situation. 
We could have $\mathcal{P}_1^{DI_2}(\Lambda)>0$, in particular $\mathcal{P}_1(\Lambda)> 0$, meaning that the channel has good and verifiable one-way private capacity. 
At the same time we could have a channel $\Lambda'$ with $\mathcal{P}_1(\Lambda)=0$ and $\mathcal{P}_2(\Lambda)>0$ that can simulate $\Lambda$ with two-way communication.
This, in particular, would mean that in order to simulate $\Lambda$ with $\Lambda'$ some key must be distilled with two-way communication.
\end{remark*}

We can then define different variants of the device-independent entanglement measures of channels, namely for a measure $\mathcal{E}(\Lambda)$ we can define different device-independent optimizations $\mathcal{E}^{\downarrow_0}$, $\mathcal{E}^{\downarrow_1}$ and $\mathcal{E}^{\downarrow_2}$ depending on the communication allowed in the device. 
Since, as an argument to the entanglement measure we are concerned with only one copy of the channel, the devices in the optimization in $\mathcal{E}^{\downarrow_j}$ are $IDI_j$ devices. 
Then, any upper bound $\mathcal{P}_{\max\{i,j\}}\leq \mathcal{E}$ leads to a device-independent upper bound $\mathcal{P}_i^{DI_j} \leq \mathcal{E}^{\downarrow_j}$.

We can finally make the same use of the partial transpose map, which we denote with $\theta$ ($\theta(\rho) = \rho^ T$).
If a channel $\Lambda$ is such that $\theta\circ\Lambda$ is also a channel, then any device for $\Lambda$ can be transformed into a device for $\theta\circ\Lambda$ with the exact same statistics;
the case of $IDI_0$ is shown in \Cref{fig:IDI0}. The consequence is the analogous of \cref{eq:pptupperbound}:
\begin{align}
    \mathcal{P}_i^{DI_j} (\Lambda) 
    \leq \mathcal{P}_{\max\{i,j\}}(\theta\circ\Lambda) 
    \leq \mathcal{P}_{2}(\theta\circ\Lambda)
       , \mathcal{E}^{\downarrow_j}(\theta\circ\Lambda),
\end{align}
which is the claim of the main text.
We would like to recall the conclusion from the main text, that it is enough to consider fully i.i.d.\ devices because the transposition map itself is i.i.d., meaning that the $n$-fold transposition is a tensor product of single channel transpositions.
Examples of channels for which $\mathcal{P}_1(\Lambda)$ is large but $\mathcal{P}_2(\theta\circ\Lambda)$ is small can be found in~\cite{christandl2017relative}.

\end{document}